\documentclass[11pt]{article}

\usepackage{amsmath, amssymb, amsfonts, amsthm}
\usepackage{todonotes}
\usepackage{bm}
\usepackage{mathtools}
\usepackage{mathrsfs}
\usepackage[left=1in,right=1in]{geometry}
\usepackage{setspace}
\usepackage[utf8x]{inputenc}
\usepackage{algpseudocode,algorithmicx,algorithm}
\usepackage[T1]{fontenc}
\usepackage{lmodern}
\usepackage{graphicx}
\usepackage{enumerate}
\usepackage{hyperref}
\usepackage{cancel}
\usepackage[all]{xy}
\usepackage{mdframed}
\usepackage[capitalize]{cleveref}

\theoremstyle{plain}
\newtheorem{thm}{Theorem}[section]
\newtheorem{lemma}[thm]{Lemma}
\newtheorem{prop}[thm]{Proposition}

\newtheorem{claim}{Claim}

\theoremstyle{definition}
\newtheorem{defn}[thm]{Definition}

\theoremstyle{remark}
\newtheorem*{rmk}{Remark}

\newcommand{\F}{\mathbb{F}}

\newcommand{\DD}{\mathcal{D}}

\newcommand{\eps}{\varepsilon}

\newcommand{\spa}{\mathrm{span}}
\newcommand{\pns}[1]{\left( #1 \right)}
\newcommand{\brk}[1]{\left[ #1 \right]}

\newcommand{\abs}[1]{\left| #1 \right|}

\newcommand{\supp}{\mathrm{supp}}
\newcommand{\poly}{\mathsf{poly}}

\renewcommand{\Pr}{\mathbf{Pr}}
\newcommand{\Prop}{\mathop{\Pr}}
\newcommand{\E}{\mathbf{E}}

\newcommand{\rank}{\mathsf{rank}}

\newcommand{\qbinom}[2]{\begin{bmatrix} {#1} \\ {#2} \end{bmatrix}_q}
\newcommand{\qbinomsmall}[2]{\brk{\begin{smallmatrix}{#1} \\ {#2}\end{smallmatrix}}_q}

\begin{document}
\title{On the List-Decodability of Random Linear Rank-Metric Codes\thanks{Research supported in part by NSF grant CCF-1422045 and NSERC grant CGSD2-502898. Some of this work was done when the first author was visiting the School of Physical and Mathematical Sciences, Nanyang Technological University, Singapore.}}
\author{Venkatesan Guruswami \and Nicolas Resch}
\date{Computer Science Department \\ Carnegie Mellon University \\ {\tt \{venkatg,nresch\}@cs.cmu.edu}}

\maketitle
\thispagestyle{empty}
\begin{abstract}
	The list-decodability of random linear rank-metric codes is shown to match that of random rank-metric codes. Specifically, an $\F_q$-linear rank-metric code over $\F_q^{m \times n}$ of rate $R = (1-\rho)(1-\frac{n}{m}\rho)-\eps$ is shown to be (with high probability) list-decodable up to fractional radius $\rho \in (0,1)$ with lists of size at most $\frac{C_{\rho,q}}{\eps}$, where $C_{\rho,q}$ is a constant depending only on $\rho$ and $q$. This matches the bound for random rank-metric codes (up to constant factors). The proof adapts the approach of Guruswami, H\aa stad, Kopparty (STOC 2010), who established a similar result for the Hamming metric case, to the rank-metric setting.
\end{abstract}

\section{Introduction}
	At its core, coding theory studies how many elements of a (finite) vector space one can pack subject to the constraint that no two elements are too close. Typically, the notion of closeness is that of Hamming distance, that is, the distance between two vectors is the number of coordinates on which they differ.  In a \emph{rank-metric} code, introduced in~\cite{delsarte78}, codewords are matrices over a finite field and the distance between codewords is the rank of their difference. A \emph{linear} rank-metric code is a subspace of matrices (over the field to which the matrix entries belong) such that every non-zero matrix in the subspace has large rank.
	
	Rank-metric codes  have found applications in magnetic recording~\cite{roth91}, public-key cryptography~\cite{GPT91,Loidreau10,Loidreau17}, and space-time coding~\cite{LGB03,LK05}. There has been a resurgence of interest in this topic due to the utility of rank-metric codes and the closely related subspace codes for error-control in random network coding~\cite{KK08,SKK08}. Decoding algorithms for rank-metric codes also have connections to the popular topic of low-rank recovery, specifically in a formulation where the task is to recover a matrix $H$ from few inner products $\langle H, M \rangle$ with measurement matrices $M$~\cite{FS12}. Finally, the study of rank-metric codes raises additional mathematical and algorithmic challenges not manifested in the Hamming metric (note that the Hamming metric case corresponds to rank-metric codes restricted to contain diagonal matrices).
	
	The notion of list-decoding, introduced independently by Elias and Wozencraft~\cite{wozencraft1958list, elias1957list}, gives the possibility of decoding past half the minimum distance of the code at the cost of returning a (hopefully small) list of candidate codewords. The goal is to determine the optimal trade-offs between the information rate, the decoding radius, and the list size. List-decoding has proved to be a highly fruitful avenue of study in the Hamming metric case, and recently there has also been a great deal of interest in the list-decodability of rank-metric codes. This work concerns the list-decodability of \emph{random linear} rank-metric codes, and establishes a trade-off between list-size and gap to optimal decoding radius that is similar to what is known (and is straightforward to establish) for completely random rank-metric codes. Almost all known constructions of rank-metric codes are linear, and random code ensembles achieve the best known trade-offs, so it is of interest to understand the performance of random linear (rank-metric) codes. The linear dependencies between sets of codewords makes such a claim non-trivial to establish in the case of linear codes. Our work is most similar to \cite{GHK11} which established a similar result for random linear codes in the Hamming metric case; we follow their overall proof strategy and adapt it to the rank-metric case.
	
	\subsection{Prior Results} \label{sec:prior_results} We now provide a summary of some previous results, before stating our result formally.
	
	\paragraph{List-decoding Gabidulin codes.} Gabidulin codes~\cite{gabidulin1985theory} provide the natural generalization of Reed-Solomon codes to the rank-metric case and have been extensively studied.
	The problem of unique decoding Gabidulin codes up to half-the-minimum-distance has been solved several times, by adapting the different approaches for unique decoding Reed-Solomon codes to the linearized setting, starting with Gabidulin's original paper, and later in \cite{roth91,loidreau06welch,KK08} among other places. Despite much effort, however, the list decoding algorithms for Reed-Solomon codes such as \cite{sudan1997decoding,GS99} haven't been generalized to Gabidulin codes. There are now results which partially explain this difficulty. 
	
	Wachter-Zeh~\cite{wachter2012bounds} has shown that there are Gabidulin codes of rate $R$ cannot be list-decoded beyond the Johnson radius $1-\sqrt{R}$, in the sense that there may be super-polynomially many Gabidulin codewords just beyond this distance from some matrix. More recently, Raviv and Wachter-Zeh~\cite{raviv2016some} (see also the correction in~\cite{raviv2017correction}) have shown that certain Gabidulin codes cannot be (combinatorially) list-decoded even slightly beyond half the minimum distance.
	
	Nonetheless, certain variants of Gabidulin codes can be list-decoded well beyond half the minimum distance. Guruswami, Wang and Xing~\cite{GWX16} (see also~\cite{GW-random14,GX-stoc13}) proved that certain explicitly constructible subcodes of the Gabidulin code of constant rate $R$ can be list-decoded up to radius $1-R-\eps$, matching the Singleton bound for rank-metric codes. These works also extend to subspace codes, a basis-independent version of rank-metric codes proposed in \cite{KK08} for error control in network coding, which spurred some of the recent interest in rank-metric codes.
	
	\paragraph{List-decoding random rank-metric codes.} The study of the list-decodability of random rank-metric codes was initiated by Ding~\cite{ding2015list}. First, she shows that a uniformly random rank-metric code in $\F_q^{m \times n}$ of rate $R$ can (with high probability) be list-decoded up to radius $1-R-\eps$ with lists of size $O(1/\eps)$, assuming $n/m \leq \eps$. Moreover, the requirement on $n/m$ is not superfluous, as if $n/m \geq \frac{2\eps}{(1-R-\eps)(R+\eps)} = \Theta_R(\eps)$, then the code cannot be list-decodable with polynomially bounded lists.
	
	For random $\F_q$-linear\footnote{Many rank-metric codes are actually $\F_{q^m}$-linear. This is done by viewing the columns of a matrix as elements of the extension field $\F_{q^m}$, so matrices $X \in \F_q^{m \times n}$ correspond to vectors $\mathbf{x} \in \F_{q^m}^n$. So a code $\mathcal{C} \subseteq \F_{q^m}^n$ is $\F_{q^m}$-linear if it is closed under multiplication by scalars from $\F_{q^m}$. However, in this paper, we focus upon $\F_q$-linear codes, so a linear code will refer to a $\F_q$-linear code.} codes, Ding shows that for any desired radius $\rho \in (0,1)$, if $R = (1-\rho)(1-\tfrac{n}{m}\rho)-\eps$, then a random linear code of rate $R$ is with high probability list-decodable with list size $\exp(O(1/\eps))$. On the negative side, if $R \geq (1-\rho)(1-\tfrac{n}{m}\rho)$, then it is shown that there are no $\F_q$-linear codes that are list-decodable up to radius $\rho$ with small lists.

 \paragraph{List-decoding random linear codes in the Hamming metric.} 
The problem of determining the list-decodability of random linear
codes in the Hamming metric remains an active area of research. As
this paper focuses upon rank-metric codes, we will not provide a
complete survey of results. However, we would like to highlight the
result of Guruswami, H{\aa}stad and Kopparty in~\cite{GHK11}, as our
approach is largely inspired by this work. The authors show that, for
any $\rho \in (0,1-\tfrac{1}{q})$, a random linear code of rate
$1-H_q(\rho)-\eps$ is list decodable up to radius $\rho$ with lists of
size $C_{q,\rho}/\eps$, for some finite constant $C_{q,\rho}$
depending only on $q$ and $\rho$. The dependence of $C_{q,\rho}$,
however, degrades badly as the error fraction $\rho$ approaches the
maximum possible value of $1-1/q$. Follow-up
works~\cite{cheraghchi2013restricted, wootters2013list,
  rudra2014every} have addressed this issue, obtaining optimal bounds
also in the high-error regime (using very different techniques).

	\subsection{Our Results} \label{sec:our_results}
	Our main result shows that random linear codes have list sizes that grow linearly with the reciprocal of the distance to capacity. 
	
	\begin{thm} \label{thm:main}
		Let $\rho \in (0,1)$ and $n\leq m$. Then, with high probability, an $\F_q$-linear rank-metric code in $\F_q^{m \times n}$ of rate $R = (1-\rho)(1-\tfrac{n}{m}\rho)-\eps$ is list-decodable up to radius $\rho$ with lists of size $O_{\rho,q}(1/\eps)$.
	\end{thm}
	
	Note that we cannot hope for a larger rate by the results in~\cite{ding2015list}. Moreover, a simple argument shows that this matches the list size which is achieved by a uniformly random code of this rate; we provide this argument in~\cref{sec:appendix}.
	
	\subsection{Organization}
	
	In~\cref{sec:prelims}, we set notation and state certain facts which we will apply. In~\cref{sec:overview} we provide the intuition for our approach before formally proving our main results in~\cref{sec:proofs}. We conclude with some open problems in~\cref{sec:conclusion}. 

\section{Preliminaries} \label{sec:prelims}

	\paragraph{Notation.} We use standard Landau notation, i.e., $O(\cdot)$, $\Omega(\cdot)$, $o(\cdot)$ and $\omega(\cdot)$. A subscript indicates that the implied constant depends on the parameter in the subscript; for example, $f(x) = O_y(g(x))$ asserts that there exists of constant $C_y$ depending on $y$ (but not $x$) such that $f(x) \leq C_yg(x)$ for all sufficiently large/small $x$.
	
	Throughout, $q$ denotes a prime power. Where convenient, we use the notation $\exp_q(\cdot) = q^{(\cdot)}$. Denote by $\F_q$ the finite field with $q$ elements, and $\F_q^{m \times n}$ the set of all $m \times n$ matrices with entries in $\F_q$, which naturally has the structure of an $\F_q$-vector space. Assume without loss of generality that $m \geq n$ (if this is not the case, consider the transpose of the matrices) and put $b = \frac{n}{m}$. For a matrix $X \in \F_q^{m \times n}$ denote its rank by $\rank(X)$. For $X,Y \in \F_q^{m \times n}$, the \emph{(normalized) rank distance} between $X$ and $Y$ is 
	\[
		d_R(X,Y) := \frac{1}{n}\rank(X-Y) \enspace .
	\]
	Observe that this indeed defines a metric (the triangle inequality is a consequence of the sub-additivity of rank). A \emph{(rank-metric) code} is then just a subset $\mathcal{C} \subseteq \F_q^{m \times n}$. If the set $\mathcal{C}$ is a subspace, it is called a \emph{linear code}. The \emph{rate} of $\mathcal{C}$ is the ratio $R := \frac{\log_q|\mathcal{C}|}{mn}$ and the \emph{minimum distance} is $d_R(\mathcal{C}) := \min\{\rank(X-Y):X,Y \in \mathcal{C}, X \neq Y\}$. Note that if $\mathcal{C}$ is linear, then $R = \frac{1}{mn}\dim_{\F_q}\mathcal{C}$ and $d_R(\mathcal{C}) = \min\{\rank(X):X \in \mathcal{C}\setminus\{0\}\}$. 
	
	As with classical codes over the Hamming metric, we have the following Singleton bound.
	
	\begin{lemma} [Singleton Bound \cite{gabidulin1985theory}]
		If $\mathcal{C} \subseteq \F_q^{m \times n}$ is a rank-metric code with minimum distance $d$, then 
		\[
			\log_q|C| \leq m(n-d+1) \enspace .
		\]
	\end{lemma}
	
	A \emph{random code} of rate $R$ is a random subset $\mathcal{C} \subseteq \F_q^{m \times n}$ obtained by including each element independently with probability $q^{-(1-R)mn}$ (thus, $\E|\mathcal{C}| = q^{Rnm}$). A \emph{random linear code} of rate $R$ is a random subspace $\mathcal{C} \subseteq \F_q^{m \times n}$ of dimension $Rmn$ (which we assume is an integer). 
	
	\paragraph{Facts about the rank-metric.} As in any metric space, we have the concept of a metric ball:
	\begin{defn} [Rank-Metric Ball]
		For $\rho \in (0,1)$ and $X \in \F_q^{m,n}$, the \emph{rank-metric ball} of radius $\rho$ centered at $X$ is 
		\[
			B_R(X,\rho) = \{Y \in \F_q^{m \times n}:d_R(X,Y) \leq \rho\} \enspace .
		\]
	\end{defn}

	Clearly, the size of a rank-metric ball depends only on its radius (and not its center). We record the following facts:
	
	\begin{lemma} [\cite{gadouleau2008decoder}] \label{lem:facts_about_balls}
		Let $N_q(r,m, n)$ denote the number of matrices in $\F_q^{m \times n}$ of rank $r$. Then
		\[
			N_q(r,m,n) = \prod_{j=0}^{r-1}\frac{(q^n-q^j)(q^m-q^j)}{q^r-q^j} \enspace .
		\]
		(By convention, we take the empty product to have value 1.) Thus, for any $X \in \F_q^{m,n}$ and $\rho \in (0,1)$:
		\[
			|B_R(X,\rho)| = \sum_{r=0}^{\lfloor \rho n \rfloor}\prod_{j=0}^{r-1}\frac{(q^n-q^j)(q^m-q^j)}{q^r-q^j} \enspace .
		\]
		Moreover, we have the estimates
		\[
			q^{mn(\rho + \rho b - \rho^2b)} \leq |B_R(X,\rho)| \leq K_q^{-1}q^{mn(\rho + \rho b - \rho^2b)} \enspace ,
		\]
		where $K_q = \prod_{j=1}^{\infty}(1-q^{-j})$. Since $K_q \in (0,1)$ increases with $q$ and $K_2 \approx 0.2887$, we have (say) $K_q^{-1} < 4$. 
	\end{lemma}

	Next, we recall the definition of the Grassmannian. 
	
	\begin{defn} [Grassmannian]
		For a vector space $V$ over $\F_q$ and an integer $0 \leq k \leq \dim V$, denote by $G(k,V)$ the set of all subspaces $U \subseteq V$ of dimension $k$. If $n = \dim V$, we have 
		\[
			|G(k,V)| = \qbinom{n}{k} = \prod_{j=0}^{k-1}\frac{q^n-q^j}{q^k-q^j} \enspace .
		\]
	\end{defn}

	We record the following estimates for $\qbinomsmall{n}{k}$:

	\begin{lemma} [\cite{gadouleau2008decoder}] \label{lem:q_nomial}
		We have 
		\[
			K_q \cdot q^{k(n-k)} \leq \qbinom{n}{k} \leq K_q^{-1}q^{k(n-k)} \enspace .
		\]		
	\end{lemma}
	
	\paragraph{List-decoding.} We now formally define list-decodability. 
	
	\begin{defn} [List-decodability]
		Let $\rho \in (0,1)$ and $L \geq 1$. A rank-metric code $\mathcal{C} \subseteq \F_q^{m \times n}$ is $(\rho,L)$ \emph{list-decodable} if for all $X \in \F_q^{m \times n}$, 
		\[
			|B_R(X,\rho) \cap \mathcal{C}| \leq L \enspace .
		\]
		If $L = \poly(n,m)$,\footnote{Here, we think of $\rho$ and $q$ as constants.} then we say that $\mathcal{C}$ is \emph{list-decodable}.
	\end{defn}
	
	\begin{rmk}
		One typically distinguishes between the combinatorial property of a code being list-decodable as defined above, vs. the algorithmic task of efficiently computing the list of all codewords near a given point. In this paper, we will only focus upon the combinatorial property of list-decodability.
	\end{rmk}
	
	\paragraph{$c$-increasing sequences.} 
	
	As in~\cite{GHK11}, the notion of a $c$-increasing sequence will be important in our proof. Recall that for $v \in \F_q^\ell$, $\supp(v) = \{i \in [\ell]:v_i \neq 0\}$. 
	
	\begin{defn} [$c$-increasing sequence] \label{def:c_increasing}
		Let $c$ be an integer. A sequence of vectors $v_1,\ldots,v_d \in \F_q^\ell$ is a \emph{$c$-increasing sequence} if for all $j \in [d]$,
		\[
			\abs{\supp(v_j)\setminus \bigcup_{i=1}^{j-1}\supp(v_i)} \geq c \enspace .
		\]
	\end{defn}

	It is shown in~\cite{GHK11} that all sets have a translate containing a large $c$-increasing sequence. A crucial ingredient in their proof was a Ramsey-theoretic lemma proved by Sauer and Shelah~\cite{sauer1972density, shelah1972combinatorial}. (More precisely, the authors use a nonstandard $q$-ary version of the Sauer-Shelah lemma.)

	\begin{lemma} [\cite{GHK11}] \label{lem:sauer-shelah}
		For every prime power $q$, and all positive integers $c,\ell$ and $L \leq q^\ell$, the following holds. For every $S \subseteq \F_q^\ell$ with $|S|=L$, there is a $w \in \F_q^\ell$ such that $S+w$ has a $c$-increasing chain of length at least 
		\[
			\frac{1}{c}\log_q\frac{L}{2} - \pns{1-\tfrac{1}{c}}\log_q((q-1)\ell) \enspace .
		\] 
	\end{lemma}

\section{Overview of Approach} \label{sec:overview}

As we show in~\cref{sec:appendix}, uniformly random codes $\mathcal{C}$ of rate $(1-\rho)(1-b\rho)-\eps$ are with high probability $(\rho,O(1/\eps))$ list-decodable. This argument is easily obtained due to the fact that, given any center $Y$ and a list $X_1,\ldots,X_L \in B(Y,\rho)$, the events ``$X_i \in \mathcal{C}$'' are independent. Hence, the probability that each $X_i$ is in the code is small enough to allow us to take a union bound over all possible lists. Unfortunately, in a uniformly linear code, the events ``$X_i \in \mathcal{C}$'' are \emph{not} independent; indeed, the events are not even 3-wise independent (as if $X_i$ and $X_j$ are in the code, then so is $X_i+X_j$). Since a list $\{X_1,\ldots,X_L\}$ is guaranteed to have a linearly independent subset of size $\log L$, one can use the argument for uniformly random codes to conclude that random linear rank-metric codes are $(\rho,O(\exp(1/\eps)))$ list-decodable -- indeed, this is more-or-less the approach followed by Ding~\cite{ding2015list}. Thus, in order to prove that lists of size $O(1/\eps)$ are sufficient, we will need to argue that, given a list contained in a small rank-metric ball which does not contain a large linearly independent set, very few elements of their span will (with high probability) also lie in the rank-metric ball. 

Such an argument is given by Guruswami, H{\aa}stad and Kopparty~\cite{GHK11}. The technical core of their argument is to show that it is exponentially unlikely that $\ell$ vectors selected uniformly at random from the Hamming ball $B_H(0,\rho) := \{x \in \F_q^n :|x|\leq \rho n \}$\footnote{Here, we use the notation $|x| := |\{i:x_i \neq 0\}|$.} have $\omega(\ell)$ elements of their linear span also lying in $B_H(0,\rho)$. That is, they show there exists a constant $C>0$ (which depends on $q$ and $\rho$) such that if $x_1,\ldots,x_\ell$ are sampled independently and uniformly at random from $B_H(0,\rho)$, the probability that $|\spa\{x_1,\ldots,x_\ell\} \cap B_H(0,\rho)| \geq C\ell$ is exponentially small in $n$. We prove an analogous result for matrices with the rank-metric in~\cref{lem:span}.

In order to achieve this, the authors first show that, for any fixed vector $y \in \F_q^n$, if one samples $x_1,x_2 \in B_H(0,\rho)$ independently and uniformly, then it is exponentially unlikely that $x_1+x_2 \in B_H(y,\rho)$. In order to bootstrap this to the case of selecting $\ell$ vectors from $B_H(0,\rho)$, the authors use~\cref{lem:sauer-shelah}.

We prove the appropriate generalization of this fact, concerning the
sum of low-rank random matrices, in~\cref{lem:two_matrices}. This
argument is a bit more involved than in~\cite{GHK11} and represents
the main technical ingredient of our paper. Once we have proved this
lemma, we are able to follow the framework of~\cite{GHK11} to conclude
our main theorem (\cref{thm:main}).

\section{Proofs} \label{sec:proofs}

As alluded to above, we begin by showing that if $X_1,X_2$ are uniformly and independently selected from $B_R(0,\rho)$, it is exponentially unlikely that $X_1+X_2 \in B_R(Y,\rho)$, where $Y$ is any fixed matrix. 

\begin{lemma} \label{lem:two_matrices}
	Let $n \leq m$ be positive integers, $Y \in \F_q^{m \times n}$ a fixed matrix, and $\rho \in (0,1)$. Let $X_1,X_2 \sim \DD_1$ denote the distribution where $X_1$ and $X_2$ are independently and uniformly selected from $B_R(0,\rho)$. Then, assuming $n,m$ are sufficiently large compared to $1-\rho$:
	\[
		\Prop_{X_1,X_2 \sim \DD_1}[X_1+X_2 \in B_R(Y,\rho)] \leq q^{-\Omega_\rho(nm)} \enspace .
	\]
\end{lemma}

Informally, the proof proceeds as follows. First, we observe that it suffices to prove that it is exponentially unlikely that $X_1+X_2 \in B_R(Y,\rho)$, where each $X_i$ is independently sampled by first choosing a subspace in $\F_q^m$ of dimension roughly $\rho n$ uniformly at random, then sampling $n$ vectors from this subspace independently and uniformly at random. We then prove that it is very unlikely that two random low-dimensional subspaces have a somewhat large intersection, cf.~\cref{helper}. By applying this claim to the orthogonal complements of the column spans of the matrices, we see that $X_1+X_2$ in this case is obtained by sampling a reasonably large subspace of $\F_q^m$ and then sampling $n$ vectors from this subspace; such a distribution has large enough support that any sample is unlikely to lie in a small rank-metric ball. 

The formal proof follows:

\begin{proof}
	Let $\Delta_1 = \Prop_{X_1,X_2 \sim \DD_1}[X_1+X_2 \in B_R(Y,\rho)]$. Let $r = \lfloor \rho n\rfloor$ and $\eps = 1-\rho>0$. We will show the probability of interest is at most $q^{-\Omega_\eps(nm)}$, since any constant depending on $\eps$ is therefore a constant depending on $\rho$. Let $s_1, s_2 \leq r$ be integers such that, conditioned on $\rank(X_1)=s_1$ and $\rank(X_2)=s_2$, the probability $\Delta_1$ is maximized. That is, the pair $(s_1,s_2)$ maximizes the expression
	\[
		\Prop_{X_1,X_2 \sim \DD_1}[X_1+X_2 \in B_R(Y,\rho)|\rank(X_j)=r_j,j=1,2] \enspace .
	\]
	Since there are at most $n^2$ choices for the pair $(s_1,s_2)$ (as they must lie in the set $\{0,1,\ldots,\lfloor \rho n\rfloor\}^2$), we have 
	\[
		\Delta_1 \leq n^2\Prop_{X_1,X_2 \sim \DD_1}[X_1+X_2\in B_R(Y,\rho)|\rank(X_j)=s_j,j=1,2] \enspace .
	\]
	Next, note that if $s_1$ or $s_2$ is $\leq (1-\delta)r$, where $\delta=\delta(\eps)>0$ is a positive constant depending on $\eps=1-\rho$ to be selected later, then since $|B_R(0,(1-\delta)\rho|/|B_R(0,\rho)| \leq q^{-\Omega_\delta(nm)}$ (cf.~\cref{lem:facts_about_balls}), we conclude 
	\[
		\Prop_{X_1,X_2 \sim \DD_1}[\rank(X_1)\leq(1-\delta)r \lor \rank(X_2) \leq (1-\delta)r] \leq q^{-\Omega_\delta(nm)} = q^{-\Omega_\rho(nm)} \enspace .
	\]
	Thus, in this case, by the total probability rule, 
	\begin{align*}
		\Delta_1 &= \sum_{(r_1,r_2)} \Prop_{X_1,X_2 \sim \DD_1}[X_1+X_2 \in B_R(Y,\rho) \land \rank(X_j)=r_j,j=1,2]\\
		&\leq n^2\Prop_{X_1,X_2 \sim \DD_1}[\rank(X_j)=s_j, j=1,2]\\
		&\leq n^2\Prop_{X_1,X_2 \sim \DD_1}[\rank(X_1)\leq(1-\delta)r \lor \rank(X_2) \leq (1-\delta)r] \\
		&\leq q^{-\Omega_\rho(nm)} \enspace .
	\end{align*}
	Hence, we now assume $(1-\delta)r \leq s_1,s_2 \leq r$. Let $\DD_2$ denote the distribution where we 
	\begin{itemize}
		\item [(a)] sample $U_1$ and $U_2$ independently and uniformly at random among all dimension $s_1$ subspaces and $s_2$ subspaces of $\F_q^m$, respectively;
		\item [(b)] sample $n$ vectors uniformly and independently from $U_1$ and put them into the columns of a matrix $X_1$, and similarly obtain $X_2$;
		\item [(c)] output the pair $(X_1,X_2)$. 
	\end{itemize}
	For $j=1,2$, under the distribution $\DD_2$ we obtain a rank $s_j$ matrix with probability at least
	\[
		(1-q^{-s_j})(1-q^{-s_j+1}) \cdots (1-q^{-2})(1-q^{-1}) \geq \prod_{j=1}^{\infty}(1-q^{-j}) \geq .288 > \frac{1}{4}
	\]
	(this is actually the probability that the first $s_j$ columns are linearly independent). Now, note that conditioned on obtaining rank $s_j$ matrices, the distributions $\DD_1$ and $\DD_2$ are identical. That is, if $E$ denotes the event that $\rank(X_j)=s_j$ for $j=1,2$, then $\DD_1|E$ and $\DD_2|E$ are identically distributed: they are both the uniform distribution over pairs of matrices $(X_1,X_2)$ with $\dim(X_j)=s_j$ for $j=1,2$. Also
	\[
		\Prop_{X_1,X_2 \sim \DD_2}[X_1+X_2 \in B_R(Y,\rho)] \geq \Prop_{X_1,X_2 \sim \DD_2}[X_1+X_2 \in B_R(Y,\rho)|E]\Prop_{X_1,X_2 \sim \DD_2}[E] \enspace ,
	\]
	so 
	\begin{align*}
		\Prop_{X_1,X_2 \sim \DD_2}[X_1+X_2 \in B_R(Y,\rho)|E] &\leq \frac{\Prop_{X_1,X_2 \sim \DD_2}[X_1+X_2 \in B_R(Y,\rho)]}{\Prop_{X_1,X_2 \sim \DD_2}[E]}\\ 
		&\leq 4^2\cdot \Prop_{X_1,X_2 \sim \DD_2}[X_1+X_2 \in B_R(Y,\rho)] \enspace .
	\end{align*}
	Recalling that 
	\[
		\Delta_1 \leq n^2 \Prop_{X_1,X_2 \sim \DD_1}[X_1+X_2 \in B_R(Y,\rho)|E] = n^2 \Prop_{X_1,X_2 \sim \DD_2}[X_1+X_2 \in B_R(Y,\rho)|E] \enspace ,
	\]
	we see that it suffices to prove
	\begin{align} \label{eq:goal}
		\Prop_{X_1,X_2 \sim \DD_2}[X_1+X_2 \in B_R(Y,\rho)] \leq q^{-\Omega_\eps(nm)} \enspace .
	\end{align}
	Towards proving~\cref{eq:goal}, we will first prove~\cref{helper}. 
	
	\begin{claim} \label{helper}
		Let $U$ and $V$ be independent and uniform subspaces of $\F_q^m$ of dimension $d_1$ and $d_2$, respectively. Suppose $d_1 \geq d_2$. Then, for any $\alpha \in (0,1)$,
		\[
			\Prop_{U,V}[\dim(U \cap V) > \alpha d_2] \leq 4^3\exp_q(\alpha(1-\alpha)d_2^2 - \alpha d_2(m-d_1)) \enspace .
		\]
	\end{claim}
	\begin{proof} [Proof of~\cref{helper}]
		Assume without loss of generality that $\alpha d_2$ is an integer. Hence, the probability of interest is
		\begin{align*}
			\Prop_{U}\brk{\exists V' \in G(\alpha d_2,V) \text{ s.t. } V' \subseteq U} &\leq \sum_{V' \in G(\alpha d_2),V}\Prop_U[V' \subseteq U] = \qbinom{d_2}{\alpha d_2} \frac{\qbinom{m-\alpha d_2}{d_1-\alpha d_2}}{\qbinom{m}{d_1}} \enspace ,
		\end{align*}
		where the equality $\Prop_{U}[V' \subseteq U] = \frac{\qbinomsmall{m-\alpha d_2}{d_1-\alpha d_2}}{\qbinomsmall{m}{d_1}}$ follows from the fact that the number of subspaces of $\F_q^m$ of dimension $d_1$ which contain a fixed subspace of dimension $\alpha d_2$ is precisely the number of subspaces of $\F_q^{m-\alpha d_2}$ of dimension $d_1-\alpha d_2$, i.e., $\qbinomsmall{m-\alpha d_2}{d_1-\alpha d_2}$. Now, using our estimates for $q$-nomial coefficients, this last quantity is at most
		\[
			K_q^{-3}\exp_q(\alpha(1-\alpha)d_2^2 + (d_1-\alpha d_2)(m-d_1) - d_1(m-d_1)) = K_q^{-3}\exp_q(\alpha(1-\alpha)d_2^2 - \alpha d_2(m-d_1)) \enspace .
		\]
		Recalling that $K_q^{-1}<4$, the claim follows.
	\end{proof}
	Now, set $\alpha = \frac{\eps^2}{\eps+\delta-\delta\eps}$ and $d_1 = n-s_1$, $d_2 = n-s_2$ in the claim (where we assume wlog that $s_1 \leq s_2$). So then 
	\[
		\alpha(1-\alpha)d_2^2 \leq \alpha d_2^2 \leq \frac{\eps^2}{\delta+\eps-\eps\delta}\pns{(\delta+\eps-\eps\delta)n}^2 = \eps^2(\eps+\delta-\eps\delta)n^2 
	\]
	and 
	\[
		-\alpha d_2(m-d_1) \geq -\frac{\eps^2}{\eps+\delta-\eps\delta}(\eps+\delta-\eps\delta)n(m-\eps n) \geq -\eps^2 nm + \eps^3n^2 \enspace .
	\]
	Now, set $\delta=\eps$. Thus, the probability that $\dim(U_1^\bot \cap U_2^\bot) > \frac{\eps^2}{\eps+\delta-\delta\eps}d_2 \geq \frac{\eps^3}{\eps+\delta-\delta\eps}n \geq \frac{\eps^2}{2}n$ is at most
	\[
		4^3\exp_q(-\eps^2 nm + \eps^3n^2 + \eps^2(2\eps-\eps^2)n^2)
	\]
	Assuming $n,m$ are sufficiently large, this is at most 
	\[
		4^3\exp_q\pns{-\tfrac{\eps^2}{2}nm} \enspace .
	\]
	We will now condition on this event not occurring. Note that this implies $\dim(U_1+U_2) \geq m - \frac{\eps^2}{2}n = (1-\frac{b\eps^2}{2})m$.
	
	Now, note that sampling $u_1 \in U_1$ and $u_2 \in U_2$ independently and uniformly at random and outputing $u_1+u_2$ is the same as sampling $v \in V := U_1+U_2$ uniformly at random. Hence, for any fixed matrix $B \in \F_q^{m \times n}$, the probability of sampling $B$ under this distribution is at most
	\[
		\exp_q\pns{-(1-\tfrac{b\eps^2}{2})m}^n =  \exp_q\pns{-nm+\tfrac{\eps^2}{2}n^2} \enspace .
	\] 
	Indeed, we need to choose $n$ vectors independently from the subspace $V$, and each vector is sampled with probability $q^{-\dim(V)} \leq q^{-(1-\tfrac{b\eps^2}{2})m}$. (Of course, if one of the columns of $Y$ is not in $V$, then we sample $Y$ with probability $0$.) Thus, the probability that we sample an element of $B_R(Y,\rho)$ if $X_1,X_2 \sim \DD_2$ and we output $X_1+X_2$, conditioned on $\dim(U_1+U_2) \geq (1-\tfrac{b\eps^2}{2})m$, is at most
	\begin{align*}
		|B_R(Y,\rho)|\exp_q\pns{-nm+\tfrac{\eps^2}{2}n^2} &\leq \exp_q\pns{(1-\eps)\eps n^2 + (1-\eps)nm - nm + \tfrac{\eps}{2}n^2}\\ 
		&\leq \exp_q\pns{-\eps nm + n^2\pns{\tfrac{\eps^2}{2} + (1-\eps)\eps}} \\
		&= \exp_q\pns{-\eps(nm-n^2)-\tfrac{\eps^2}{2}n^2}\enspace .
	\end{align*}
	Note that either if $m = \omega(n)$ or $m = \Theta(n)$, we have that the term in the exponent is $-\Theta_\eps(nm)$. This establishes~\cref{eq:goal} and therefore completes the proof. 
\end{proof}

We now show that if $\ell$ matrices from $B_R(0,\rho)$ are chosen at random, then it is unlikely that $\omega(\ell)$ of their linear combinations lie in $B_R(0,\rho)$. The proof combines Lemmas~\ref{lem:sauer-shelah} and~\ref{lem:two_matrices}.

\begin{lemma} \label{lem:span}
	For every $\rho \in (0,1)$, there is a constant $C=C_{\rho,q}>1$ such that for all integers $n \leq m$ and $\ell = o(\sqrt{nm})$, if $X_1,\ldots,X_\ell$ are selected independently and uniformly at random from $B_R(0,\rho)$, then
	\[
		\Prop[|\spa\{X_1,\ldots,X_\ell\}\cap B_R(0,\rho)| \geq C\cdot \ell] \leq q^{-(4-o(1))nm} \enspace .
	\]
\end{lemma}

\begin{proof}
	Let $L = C \cdot \ell$ (for some $C=C_{\rho,q}$ to be selected later) and let $c=2$. Let $\delta=\delta_{\rho}$ be the constant in the $\Omega_\rho(\cdot)$ from~\cref{lem:two_matrices}. Let 
	\begin{align*}
		d &= \bigg\lfloor \frac{1}{c}\log_q\frac{L}{2} - \pns{1-\frac{1}{c}}\log_q((q-1)\ell)\bigg\rfloor = \bigg\lfloor \frac{1}{2}\log_q\frac{L}{2} - \frac{1}{2}\log_q((q-1)\ell)\bigg\rfloor\\ 
		&\geq \frac{1}{2}\log_q\frac{L}{2(q-1)\ell} - 1 = \frac{1}{2}\log_q\frac{C}{2(q-1)q^2} \enspace .
	\end{align*}
	Finally, for a vector $u \in \F_q^\ell$, let $X(u) = \sum_{i}u_iX_i$. 
	
	Towards proving the lemma, we prove the following claim:
	
	\begin{claim} \label{helper2}
		For any $S \subseteq \F_q^\ell$ with $|S| = L+1$, 
		\begin{align}\label{eq:bound}
			\Prop[\forall v \in S, X(v) \in B_R(0,\rho)] < q^{nm}q^{-\delta dnm} \enspace .
		\end{align}
	\end{claim}
	
	\begin{proof} [Proof of~\cref{helper2}]
		Let $w$ and $v_1,\ldots,v_d \in S$ be as given by~\cref{lem:sauer-shelah}. That is, $v_1+w,v_2+w,\ldots,v_d+w$ is a 2-increasing sequence. Then
		\begin{align*}
			\Pr[\forall v \in S, X(v) \in B_R(0,\rho)] &\leq \Pr[\forall j \in [d],X(v_j) \in B_R(0,\rho)]\\
			&= \Pr[\forall j \in [d], X(v_j)+X(w) \in B_R(X(w),\rho)]\\
			&= \Pr[\forall j \in [d], X(v_j+w) \in B_R(X(w),\rho)]
		\end{align*}
		Fix $Y \in \F_q^{m \times n}$. Then
		\begin{align*}
			\Pr&[\forall j \in [d], X(v_j+w) \in B_R(Y,\rho)]\\
			&= \prod_{j=1}^d\Pr[X(v_j+w) \in B_R(Y,\rho)|X(v_i+w) \in B_R(Y,\rho) ~\forall 1 \leq i \leq j-1]\\
			&\leq \prod_{j=1}^d\max_{\substack{Z_k \in B_R(0,\rho): \\ k \in \bigcup_{i=1}^{j-1}\supp(v_i+w)}}\Pr\brk{X(v_j+w) \in B_R(Y,\rho)|X_k=Z_k ~ \forall k \in \bigcup_{i=1}^{j-1}\supp(v_i+w)}\\
			&\leq \pns{q^{-\delta nm}}^d \enspace .
		\end{align*}
		The last inequality follows from~\cref{lem:two_matrices} as follows: let $i_1,i_2$ be distinct elements of $\supp(v_j+w) \setminus \bigcup_{i=1}^{j-1}\supp(v_i+w)$ (which exist thanks to the 2-increasing property). Then apply~\cref{lem:two_matrices} with $A_1 = (v_j)_{i_1}X_{i_1}$, $A_2 = (v_j)_{i_2}X_{i_2}$ (which are distributed uniformly over $B_R(0,\rho)$ if $X_{i_1},X_{i_2}$ are), and $B = Y - \sum_{k \in [\ell]\setminus\{i_1,i_2\}}(v_j+w)_kX_k = Y - \sum_{k \in [\ell]\setminus\{i_1,i_2\}}(v_j+w)_kZ_k$ (which is a fixed matrix). 
		
		By taking a union bound over all $q^{nm}$ choices of $Y \in \F_q^{m \times n}$, the claim follows. 
	\end{proof}
	We now bound the probability that more than $L$ elements of $\spa\{X_1,\ldots,X_\ell\}$ lie in $B_R(0,\rho)$. This occurs iff there exists a subset $S \subseteq \F_q^\ell$ of size $L+1$ such that $\forall v \in S$, $X(v) \in B_R(0,\rho)$. By taking a union bound over the probability in \eqref{eq:bound}, this occurs with probability at most $q^{\ell(L+1)}q^{nm}q^{-\delta dnm}$. Assuming $C=C_{\rho,q}$ is large enough so that $d \geq \frac{5}{\delta}$, this probability is at most
	\[
		q^{o(nm)+nm-5nm} = q^{-(4-o(1))nm} \enspace . \qedhere
	\]
\end{proof}

We are now prepared to prove~\cref{thm:main}, which we now restate formally. 

\begin{thm} [\cref{thm:main}, restated]
	Let $\rho \in (0,1)$, $n\leq m$ integers and set $b = \frac{n}{m}$. Then there exists a constant $c=c_{\rho,q}>0$ such that for any $\eps>0$ and sufficiently large $n,m$, letting $R = (1-\rho)(1-b\rho)-\eps$, if $\mathcal{C} \subseteq \F_q^{m \times n}$ is a random linear code of rate $R$, then
	\[
		\Pr[\mathcal{C} \text{ is } (\rho,\tfrac{c}{\eps}) \text{ list-decodable}] > 1-q^{-nm} \enspace .
	\]
\end{thm}

\begin{proof}
	Let $c = 2C$, where $C$ is the constant from~\cref{lem:span}, let $L = \lceil\tfrac{c}{\eps}\rceil$, and let $n,m$ be larger than $L$ and sufficiently large so that the $o(1)$ term of~\cref{lem:span} is at most 1. 
	
	For $X \in \F_q^{m \times n}$ selected uniformly at random, we will study the quantity 
	\[
		\Delta := \Prop_{\mathcal{C},X}[|B_R(X,\rho)\cap \mathcal{C}| \geq L] \enspace .
	\]
	By taking a union bound over $X$, note that proving $\Delta \leq q^{-nm} \cdot q^{-nm}$ will suffice to conclude the theorem. 
	
	As a first step, we show that we can move $X$ to the origin without significantly changing the probability $\Delta$. Indeed, 
	\begin{align*}
		\Delta 	&= \Prop_{\mathcal{C},X}[|B_R(X,\rho)\cap \mathcal{C}| \geq L]\\
				&= \Prop_{\mathcal{C},X}[|B_R(0,\rho)\cap \mathcal{C}+X| \geq L]\\
				&\leq \Prop_{\mathcal{C},X}[|B_R(0,\rho)\cap \mathcal{C}+\{0,X\}| \geq L]\\
				&\leq \Prop_{\mathcal{C}^*}[|B_R(0,\rho)\cap \mathcal{C}^*| \geq L] \enspace ,
	\end{align*}
	where $\mathcal{C}^*$ is a random $Rnm+1$ dimensional subspace containing $\mathcal{C} + \{0,X\}$. More explicitly, $\mathcal{C^*}$ is sampled by first sampling a dimension $Rnm$ subspace $\mathcal{C} \subseteq \F_q^{mn}$. Then, if $X \notin \mathcal{C}$, we set $\mathcal{C}^* = \mathcal{C}+\{\alpha X: \alpha \in \F_q\}$; while if $X \in \mathcal{C}$, we set $\mathcal{C}^* = \mathcal{C} + \{\alpha Y:\alpha \in \F_q\}$ where $Y$ is picked uniformly at random from $\F_q^{mn} \setminus \mathcal{C}$. Recalling that $X$ is uniformly random, we see that $\mathcal{C}^*$ is a uniformly random subspace of dimension $Rnm+1$. 
	
	Now, for each integer $\ell$ satisfying $\log_q L \leq \ell \leq L$, let $\mathcal{F}_\ell$ denote the set of all tuples $(A_1,\ldots,A_\ell) \in B_R(0,r)^\ell$ such that $A_1,\ldots,A_\ell$ are linearly independent and $|\spa\{A_1,\ldots,A_\ell\} \cap B_R(0,\rho)| \geq L$. Let 
	\[
		\mathcal{F} = \bigcup_{\log_qL \leq \ell \leq L}\mathcal{F}_\ell \enspace .
	\]
	Denote $\mathbf{A} = (A_1,\ldots,A_\ell)$ and $\{\mathbf{A}\} = \{A_1,\ldots,A_\ell\}$ (i.e., $\mathbf{A}$ denotes the ordered tuple whereas $\{\mathbf{A}\}$ denotes the unordered set). 
	
	Towards bounding $\Prop_{\mathcal{C}^*}[|B_R(0,\rho)\cap \mathcal{C}^*| \geq L]$, notice that if $|B_R(0,\rho)\cap \mathcal{C}^*| \geq L$, then there must exist some $\mathbf{A} \in \mathcal{F}$ for which $\mathcal{C}^* \supseteq \{\mathbf{A}\}$. Indeed, we may choose any maximal linearly independent subset of $B_R(0,\rho) \cap \mathcal{C}^*$ if this set has size $\leq L$, or any linearly independent subset of $B_R(0,\rho) \cap \mathcal{C}^*$ of size $L$ otherwise. 
	
	Thus, by a union bound,
	\[
		\Delta \leq \sum_{\mathbf{A} \in \mathcal{F}}\Prop_{\mathcal{C}^*}[\mathcal{C}^* \supseteq \{\mathbf{A}\}] = \sum_{\ell = \lceil\log_q L\rceil}^{L}\sum_{\mathbf{A} \in \mathcal{F}_\ell}\Prop_{\mathcal{C}^*}[\mathcal{C}^* \supseteq \{\mathbf{A}\}] \enspace .
	\]
	Note that for $\mathbf{A} = (A_1,\ldots,A_\ell) \in \mathcal{F}$, by linear independence we have
	\[
		\Prop_{\mathcal{C}^*}[\mathcal{C}^* \supseteq \{\mathbf{A}\}] = \prod_{j=1}^\ell \Prop_{\mathcal{C}^*}[A_j\in \mathcal{C}^*|A_1,\ldots,A_{j-1} \in \mathcal{C}^*] = \prod_{j=1}^{\ell}\frac{q^{Rnm+1}-q^{j-1}}{q^{nm}-q^{j-1}} \leq \pns{\frac{q^{Rnm+1}}{q^{nm}}}^\ell \enspace .
	\]
	Thus, we find
	\[
		\Delta \leq \sum_{\ell = \lceil\log_q L\rceil}^{L}|\mathcal{F}_\ell|\cdot\pns{\frac{q^{Rnm+1}}{q^{nm}}}^\ell \enspace .
	\]
	We now bound $|\mathcal{F}_\ell|$ depending on the value of $\ell$. 
	
	\begin{itemize}
		\item \textbf{Case 1.} $\ell < \frac{3}{\eps}$. 
		
		In this case, note that $\frac{|\mathcal{F}_\ell|}{|B_R(0,\rho)|^\ell}$ is a lower bound on the probability that $\ell$ matrices $X_1,\ldots,X_\ell$ chosen independently and uniformly at random from $B_R(0,\rho)$ are such that 
		\[
			|\spa\{X_1,\ldots,X_\ell\} \cap B_R(0,\rho)| \geq L \enspace.
		\]
		\cref{lem:span} tells us that this probability is at most $q^{-3nm}$. Thus,
		\[
			|\mathcal{F}_\ell| \leq |B_R(0,\rho)|^\ell q^{-3nm} \leq \pns{4q^{mn(\rho+\rho b-\rho^2b)}}^\ell \cdot q^{-3nm}
		\]
		
		\item \textbf{Case 2.} $\ell \geq \frac{3}{\eps}$. 
		
		In this case, we have the (simple) bound of 
		\[
			|\mathcal{F}_\ell| \leq |B_R(0,\rho)|^\ell \leq \pns{4q^{mn(\rho+\rho b-\rho^2b)}}^\ell \enspace .
		\]
	\end{itemize}
	Combining these inequalities, we obtain the following bound:
	\begin{align*}
		\Delta &\leq \sum_{\ell = \lceil \log_q L\rceil}^{\lceil \tfrac{3}{\eps}\rceil-1}|\mathcal{F}_\ell|\cdot\pns{\frac{q^{Rnm+1}}{q^{nm}}}^\ell + \sum_{\ell=\lceil \tfrac{3}{\eps}\rceil}^{L}|\mathcal{F}_\ell|\cdot\pns{\frac{q^{Rnm+1}}{q^{nm}}}^\ell\\
		&\leq \sum_{\ell = \lceil \log_q L\rceil}^{\lceil \tfrac{3}{\eps}\rceil-1} \pns{4q^{mn(\rho+\rho b-\rho^2b)}}^\ell \cdot q^{-3nm} \cdot \pns{\frac{q^{Rnm}}{q^{nm}}}^\ell\cdot q^\ell + \sum_{\ell=\lceil \tfrac{3}{\eps}\rceil}^{L}\pns{4q^{mn(\rho+\rho b-\rho^2b)}}^\ell\cdot\pns{\frac{q^{Rnm}}{q^{nm}}}^\ell\cdot q^\ell\\
		&\leq q^{-3nm}\sum_{\ell = \lceil \log_q L\rceil}^{\lceil \tfrac{3}{\eps}\rceil-1}4^\ell\cdot q^\ell\cdot q^{(-\eps nm)\ell} + \sum_{\ell=\lceil \tfrac{3}{\eps}\rceil}^{L}4^\ell\cdot q^\ell \cdot q^{(-\eps nm)\ell}\\
		&\leq (4q)^L\pns{q^{-3nm} \cdot\frac{3}{\eps} + L\cdot q^{-\eps nm \cdot \frac{3}{\eps}}}\\
		&< q^{-nm} \cdot q^{-nm}
	\end{align*}
	assuming $n,m$ are large enough compared to $\eps$. 
\end{proof}

\section{Conclusion} \label{sec:conclusion}

We have shown that random $\F_q$-linear rank-metric codes of rate $R =
(1-\rho)(1-b\rho)-\eps$ are with high probability $(\rho,O(1/\eps))$
list-decodable, where the big-$O$ notation hides constants depending
only on $\rho$ and $q$. This matches the performance of uniformly
random rank-metric codes up to constant factors.

Many open directions remain to be pursued; we mention a couple of
problems that we find particularly interesting. First of all, we are
unable to give good control of the list size when $\rho \to 1$. One
can show that if $\rho = 1-\eps$, then there exist codes of rate
$\Omega(\eps-\eps b + \eps^2b)$ which are $O(1/(\eps-\eps b +
\eps^2b))$ list-decodable. We provide this argument
in~\cref{sec:appendix}. Proving that linear codes can achieve a
similar tradeoff remains an interesting open problem. We remark that
similar issues with the proof of \cite{GHK11} for the high noise
regime in the Hamming metric case were addressed and resolved, using
different techniques (based on appropriate Gaussian processess) in
\cite{cheraghchi2013restricted,wootters2013list,rudra2014every}. A recent
work~\cite{rudra2017average} provides a common proof for all noise
regimes albeit with weaker list size guarantees. It will be
interesting to see if these other approaches can be adapted to the
rank-metric setting.

Lastly, we note that it is common to view a rank-metric code
$\mathcal{C}$ as a subset of $\F_{q^m}^n$, and then insist that such a
code be $\F_{q^m}$-linear. This is done by fixing a basis for
$\F_{q^m}$ over $\F_q$ and then identifying a vector $\mathbf{x} \in
\F_{q^m}^n$ with the matrix $X \in \F_q^{m \times n}$, where the $i$th
column of $X$ is $\mathbf{x}_i$ written in the coordinates defined by
the basis. Thus, it is natural to ask if a random $\F_{q^m}$-linear
subspace $\mathcal{C} \subset \F_{q^m}^n$ is rank-metric
list-decodable. By adjusting the constant $C$ in the proof
of~\cref{lem:span}, one can see that the proof still goes
through. Unfortunately, $C$ will have to grow polynomially in $q^m$
(rather than just $q$), so the resulting list sizes will be on the
order of $q^{O(m)}/\eps$. Thus, we are unable to conclude that random
$\F_{q^m}$-linear codes are rank-metric list-decodable, let alone
prove the optimal $O(1/\eps)$ list size. Indeed, we are currently
unaware of a proof that \emph{any} $\F_{q^m}$-linear rank-metric codes
are list-decodable beyond half the minimum distance (the codes
constructed by Guruswami, Wang and Xing~\cite{GWX16}
do not satisfy this property). Thus, existentially proving that some
$\F_{q^m}$-linear rank-metric code is list-decodable or concluding
that no such code exists would represent an important step forward in
our understanding of the list-decodability of rank-metric codes.

\bibliographystyle{alpha}
\bibliography{refs}

\begin{thebibliography}{GWX16}

\bibitem[CGV13]{cheraghchi2013restricted}
Mahdi Cheraghchi, Venkatesan Guruswami, and Ameya Velingker.
\newblock Restricted isometry of {F}ourier matrices and list decodability of
  random linear codes.
\newblock {\em SIAM Journal on Computing}, 42(5):1888--1914, 2013.

\bibitem[Del78]{delsarte78}
Philippe Delsarte.
\newblock Bilinear forms over a finite field, with applications to coding
  theory.
\newblock {\em J. Comb. Theory, Ser. A}, 25(3):226--241, 1978.

\bibitem[Din15]{ding2015list}
Yang Ding.
\newblock On list-decodability of random rank metric codes and subspace codes.
\newblock {\em IEEE Transactions on Information Theory}, 61(1):51--59, 2015.

\bibitem[Eli57]{elias1957list}
Peter Elias.
\newblock {\em List decoding for noisy channels}.
\newblock Research Laboratory of Electronics, Massachusetts Institute of
  Technology, 1957.

\bibitem[FS12]{FS12}
Michael~A. Forbes and Amir Shpilka.
\newblock On identity testing of tensors, low-rank recovery and compressed
  sensing.
\newblock In {\em Proceedings of the 44th ACM Symposium on Theory of Computing
  Conference}, pages 163--172, 2012.

\bibitem[Gab85]{gabidulin1985theory}
Ernest~M. Gabidulin.
\newblock Theory of codes with maximum rank distance.
\newblock {\em Problemy Peredachi Informatsii}, 21(1):3--16, 1985.

\bibitem[GHK11]{GHK11}
Venkatesan Guruswami, Johan H{\aa}stad, and Swastik Kopparty.
\newblock On the list-decodability of random linear codes.
\newblock {\em {IEEE} Trans. Information Theory}, 57(2):718--725, 2011.
\newblock Prelim. version in STOC 2010.

\bibitem[GPT91]{GPT91}
Ernst~M. Gabidulin, A.~V. Paramonov, and O.~V. Tretjakov.
\newblock Ideals over a non-commutative ring and their applications in
  cryptology.
\newblock In {\em EUROCRYPT}, pages 482--489, 1991.

\bibitem[GS99]{GS99}
Venkatesan Guruswami and Madhu Sudan.
\newblock Improved decoding of {R}eed-{S}olomon and algebraic-geometry codes.
\newblock {\em {IEEE} Trans. Information Theory}, 45(6):1757--1767, 1999.

\bibitem[GW14]{GW-random14}
Venkatesan Guruswami and Carol Wang.
\newblock Evading subspaces over large fields and explicit list-decodable
  rank-metric codes.
\newblock In {\em Approximation, Randomization, and Combinatorial Optimization.
  Algorithms and Techniques, {APPROX/RANDOM} 2014, September 4-6, 2014,
  Barcelona, Spain}, pages 748--761, 2014.

\bibitem[GWX16]{GWX16}
Venkatesan Guruswami, Carol Wang, and Chaoping Xing.
\newblock Explicit list-decodable rank-metric and subspace codes via subspace
  designs.
\newblock {\em IEEE Transactions on Information Theory}, 62(5):2707--2718,
  2016.

\bibitem[GX13]{GX-stoc13}
Venkatesan Guruswami and Chaoping Xing.
\newblock List decoding {R}eed-{S}olomon, algebraic-geometric, and {G}abidulin
  subcodes up to the {S}ingleton bound.
\newblock In {\em Proceedings of the forty-fifth annual ACM symposium on Theory
  of computing}, pages 843--852. ACM, 2013.

\bibitem[GY08]{gadouleau2008decoder}
Maximilien Gadouleau and Zhiyuan Yan.
\newblock On the decoder error probability of bounded rank-distance decoders
  for {M}aximum {R}ank {D}istance codes.
\newblock {\em IEEE Transactions on Information Theory}, 54(7):3202--3206,
  2008.

\bibitem[KK08]{KK08}
Ralf Koetter and Frank~R. Kschischang.
\newblock Coding for errors and erasures in random network coding.
\newblock {\em {IEEE} Trans. Information Theory}, 54(8):3579--3591, 2008.

\bibitem[LGB03]{LGB03}
P.~Lusina, Ernst~M. Gabidulin, and Martin Bossert.
\newblock Maximum rank distance codes as space-time codes.
\newblock {\em {IEEE} Trans. Information Theory}, 49(10):2757--2760, 2003.

\bibitem[LK05]{LK05}
Hsiao{-}feng Lu and P.~Vijay Kumar.
\newblock A unified construction of space-time codes with optimal
  rate-diversity tradeoff.
\newblock {\em {IEEE} Trans. Information Theory}, 51(5):1709--1730, 2005.

\bibitem[Loi06]{loidreau06welch}
Pierre Loidreau.
\newblock A {W}elch--{B}erlekamp like algorithm for decoding {G}abidulin codes.
\newblock In {\em Coding and cryptography}, pages 36--45. Springer, 2006.

\bibitem[Loi10]{Loidreau10}
Pierre Loidreau.
\newblock Designing a rank metric based {M}celiece cryptosystem.
\newblock In {\em Post-Quantum Cryptography (PQCrypto) - Third International
  Workshop}, pages 142--152, 2010.

\bibitem[Loi17]{Loidreau17}
Pierre Loidreau.
\newblock A new rank metric codes based encryption scheme.
\newblock In {\em Post-Quantum Cryptography (PQCrypto) - 8th International
  Workshop}, pages 3--17, 2017.

\bibitem[Rot91]{roth91}
Ron~M. Roth.
\newblock Maximum-rank array codes and their application to criss-cross error
  correction.
\newblock {\em {IEEE} Trans. Information Theory}, 37(2):328--336, 1991.

\bibitem[RW14]{rudra2014every}
Atri Rudra and Mary Wootters.
\newblock Every list-decodable code for high noise has abundant near-optimal
  rate puncturings.
\newblock In {\em Proceedings of the forty-sixth annual ACM symposium on Theory
  of computing}, pages 764--773. ACM, 2014.

\bibitem[RW17]{rudra2017average}
Atri Rudra and Mary Wootters.
\newblock Average-radius list-recovery of random linear codes: it really ties
  the room together.
\newblock {\em arXiv preprint arXiv:1704.02420}, 2017.

\bibitem[RWZ16]{raviv2016some}
Netanel Raviv and Antonia Wachter-Zeh.
\newblock Some {G}abidulin codes cannot be list decoded efficiently at any
  radius.
\newblock {\em IEEE Transactions on Information Theory}, 62(4):1605--1615,
  2016.

\bibitem[RWZ17]{raviv2017correction}
Netanel Raviv and Antonia Wachter-Zeh.
\newblock A correction to “{S}ome {G}abidulin codes cannot be list decoded
  efficiently at any radius”.
\newblock {\em IEEE Transactions on Information Theory}, 63(4):2623--2624,
  2017.

\bibitem[Sau72]{sauer1972density}
Norbert Sauer.
\newblock On the density of families of sets.
\newblock {\em Journal of Combinatorial Theory, Series A}, 13(1):145--147,
  1972.

\bibitem[She72]{shelah1972combinatorial}
Saharon Shelah.
\newblock A combinatorial problem; stability and order for models and theories
  in infinitary languages.
\newblock {\em Pacific Journal of Mathematics}, 41(1):247--261, 1972.

\bibitem[SKK08]{SKK08}
Danilo Silva, Frank~R. Kschischang, and Ralf Koetter.
\newblock A rank-metric approach to error control in random network coding.
\newblock {\em {IEEE} Trans. Information Theory}, 54(9):3951--3967, 2008.

\bibitem[Sud97]{sudan1997decoding}
Madhu Sudan.
\newblock Decoding of reed solomon codes beyond the error-correction bound.
\newblock {\em Journal of complexity}, 13(1):180--193, 1997.

\bibitem[Woo13]{wootters2013list}
Mary Wootters.
\newblock On the list decodability of random linear codes with large error
  rates.
\newblock In {\em Proceedings of the forty-fifth annual ACM symposium on Theory
  of computing}, pages 853--860. ACM, 2013.

\bibitem[Woz58]{wozencraft1958list}
John~M Wozencraft.
\newblock List decoding.
\newblock {\em Quarterly Progress Report}, 48:90--95, 1958.

\bibitem[WZ12]{wachter2012bounds}
Antonia Wachter-Zeh.
\newblock Bounds on list decoding {G}abidulin codes.
\newblock {\em arXiv preprint arXiv:1205.0345}, 2012.

\end{thebibliography}

\appendix

\section{Existential Results for Random Codes} \label{sec:appendix}

We now provide certain existential arguments in order to set expectations. These arguments are completely standard (and presumably have appeared elsewhere; indeed,~\cref{prop:random_fixed_radius} is more-or-less implicit in~\cite{ding2015list}). Recall that a random code $\mathcal{C}$ of rate $R \in (0,1)$ is sampled by including each element $X \in \F_q^{m \times n}$ in $\mathcal{C}$ with probability $q^{(R-1)mn}$. (Thus, $\E[|\mathcal{C}|] = q^{Rmn}$.)

\paragraph{Random Codes of Fixed Radius} First, we show that random codes achieve the same parameters as we have shown random linear codes achieve. This 

\begin{prop} \label{prop:random_fixed_radius}
	Let $\eps>0$ and $\rho \in (0,1)$. A random code $\mathcal{C}$ of rate $R := (1-\rho)(1-b\rho)-\eps$ is $(\rho,O(1/\eps))$-list-decodable with probability at least $1-q^{-\Theta(mn)}$, assuming $m,n$ are sufficiently large compared to $1/\eps$. 
\end{prop}

\begin{proof}
	$\mathcal{C}$ is not $(\rho,L)$ list-decodable if there exists a center $Y \in \F_q^{m \times n}$ and a list $\{X_1,\ldots,X_{L+1}\} \subset B(Y,\rho)$ such that $X_i \in \mathcal{C}$ for all $i=1,\ldots,L+1$. Using a union bound, the independence of the events $X_i \in \mathcal{C}$, estimates from~\cref{lem:facts_about_balls} and the definition of $R$, we find:
	\begin{align*}
		\Pr[\exists Y \in \F_q^{m \times n} \text{ and }& \{X_1,\ldots,X_{L+1}\} \subset B(Y,\rho) \text{ s.t. }X_i \in \mathcal{C} ~ \forall i=1,\ldots,L+1]\\
		&\leq \sum_{Y \in \F_q^{m \times n}}\sum_{\{X_1,\ldots,X_{L+1}\} \subset B(Y,\rho)}\prod_{i=1}^{L+1}\Pr[X_i \in \mathcal{C}]\\
		&\leq q^{mn}\binom{4q^{mn(\rho+\rho b-\rho^2 b)}}{L+1} q^{(R-1)mn(L+1)}\\
		&\leq 4^{L+1}\exp_q\pns{mn+mn(\rho+\rho b-\rho^2 b)(L+1) + (-\rho-\rho b+\rho^2b-\eps)(L+1)}\\
		&= 4^{L+1}\exp_q\pns{mn - \eps mn(L+1)} \enspace .
	\end{align*}
	Hence, by setting $L = \Theta(1/\eps)$, assuming $m,n$ are large enough compared to $1/\eps$, the previous expression is $q^{-\Theta(mn)}$, as desired. 
\end{proof}

\paragraph{Random Codes in the Large Radius Regime}

Now we imagine that the decoding radius is tending to 1. In this case, we show that random codes of rate $\Omega(\eps-\eps b + \eps^2b)$ are $(1-\eps,O(\tfrac{1}{\eps-\eps b + \eps^2b}))$ list-decodable. We note that, in the special case of $b=1 \iff n=m$, we see that random codes of rate $\Omega(\eps^2)$ are $(1-\eps,O(1/\eps^2))$ list-decodable with high probability.

\begin{prop} \label{prop:large_radius}
	Let $\eps>0$. Then a random code $\mathcal{C}$ of rate $R = \tfrac{\eps-\eps b+\eps^2b}{2}$ is $(1-\eps,\tfrac{4}{\eps-\eps b + \eps^2 b})$ list-decodable with probability $1-q^{-\Theta(mn)}$ for sufficiently large $m,n$.  
\end{prop}

\begin{proof}
	Let $\rho = 1-\eps$ and let $L = \lceil\frac{4}{\eps - \eps b + \eps^2b}\rceil-1$. As before, we bound
	\begin{align*}
		\Pr[\exists Y \in \F_q^{m \times n} \text{ and }& \{X_1,\ldots,X_{L+1}\} \subset B(Y,\rho) \text{ s.t. }X_i \in \mathcal{C} ~ \forall i=1,\ldots,L+1]\\
		&\leq \sum_{Y \in \F_q^{m \times n}}\sum_{\{X_1,\ldots,X_{L+1}\} \subset B(Y,\rho)}\prod_{i=1}^{L+1}\Pr[X_i \in \mathcal{C}]\\
		&\leq q^{mn}\binom{4q^{mn(\rho+\rho b-\rho^2 b)}}{L+1} q^{(R-1)mn(L+1)}\\
		&\leq 4^{L+1} \exp_q\pns{mn + mn(L+1)((1-\eps)+(1-\eps)b-(1-\eps)^2b - \tfrac{\eps-\eps b+\eps^2b}{2})}\\
		&\leq 4^{L+1}\exp_q\pns{mn + mn\tfrac{4}{\eps-\eps b + \eps^2b}\pns{\tfrac{\eps-\eps b + \eps^2b}{2}}}\\
		&= 4^{L+1}\exp_q(-mn) \enspace .
	\end{align*}
	Assuming $m,n$ are large enough compared to $1/\eps$, the previous expression is $q^{-\Theta(mn)}$, as desired. 
\end{proof}

\end{document}